%% file: root.tex
\documentclass[letterpaper, 10 pt, conference]{ieeeconf} 

\usepackage{amsmath}
\usepackage{pgfplots}
\usepackage{bbold}
\usepackage[english]{babel}
\usepackage{subcaption}
\usepackage{url}
\usepackage{mathtools}

\newcommand{\karmoose}{\color{magenta}}

\newtheorem{theorem}{Theorem}[section]
\newtheorem{lemma}[theorem]{Lemma}
\newtheorem{corollary}[theorem]{Corollary}

\newcommand{\Set}{\mathcal}

\IEEEoverridecommandlockouts 

\begin{document}
\title{\LARGE \bf Distorting an Adversary's View in Cyber-Physical Systems}
\author{
{Gaurav Kumar Agarwal, Mohammed Karmoose, Suhas Diggavi, Christina Fragouli, Paulo Tabuada} \\
	Department of Electrical and Computer Engineering, UCLA, Los Angeles, USA\\
	Email: \{gauravagarwal, mkarmoose, suhasdiggavi, christina.fragouli, tabuada\}@ucla.edu
	\thanks{The work was partially supported by NSF grant 1740047, by the Army Research Laboratory under Cooperative Agreement W911NF-17-2-0196, and by the UC-NL grant LFR-18-548554.
	}
}
\maketitle
\IEEEoverridecommandlockouts
\thispagestyle{empty}
\pagestyle{empty}

\input{Abstract_V3GA}
\input{Intro_V7MK}

\input{SystemModel_V9MK}

\input{AverageDist_V10MK}

\input{WorstDist_V12MK}

\input{Appendices_V3GA}

\bibliographystyle{IEEEtran}
\bibliography{CDC2018}

%

\end{document}

%% file: Abstract_V3GA.tex
\begin{abstract}
In Cyber-Physical Systems (CPSs), inference based on communicated data is of critical significance as it can be used to manipulate or damage the control operations by adversaries. This calls for efficient mechanisms for secure transmission of data since control systems are becoming increasingly distributed over larger geographical areas. Distortion based security, recently proposed as one candidate for CPSs security, is not only more appropriate for these applications but also quite frugal in terms of prior requirements on shared keys. In this paper, we propose distortion-based metrics to protect CPSs communication and show that it is possible to confuse adversaries with just a few bits of pre-shared keys.
\end{abstract}

%% file: Intro_V7MK.tex
\section{Introduction}
It is well recognized that wireless networking is essential to realize the potential of new CPS applications, and is equally well recognized that  private and secure exchange of information is a necessary and not simply a desirable condition for the CPS ecosystem to thrive. For instance, personal health data in assisted environments, car positions and trajectories, proprietary interests, all need to be protected. We introduces a new approach to CPS security, that aims to distort an adversary's view of a control system's states.

Our starting observation is that information security measures (cryptographic and information theoretic secrecy), are not well matched to  CPS applications as they
impose unnecessary requirements, such as protecting all the raw data, and thus can cause high operational costs. Cryptographic methods rely on computational complexity: they require short keys, but high complexity at the communicating nodes (that can be simple sensors in some cases), and can impose a significant overhead on short packet transmissions, therefore increasing delay~\cite{wan2016exploiting,Zan2013KeyAA,Trappe,6521318}. Information theoretic methods rely on  keys: they have low complexity and do not add packet overhead, but require the communicating nodes to share large keys - every communication link needs to use a shared secret key (for a one-time pad) of length equal to the entropy (effectively length) of the transmitted data~\cite{shannon1949communication}. These costs accumulate rapidly given that large CPS applications can have dense communication patterns. 

Instead, we propose a lightweight approach, that uses small amounts of key and low complexity operations, and builds around a distortion measure. To illustrate\footnote{Although we illustrate our approach for a specific simple example, it extends to protecting general system states.}, consider  the following simple example of a drone flying motion inside a square, depicted in Fig.~\ref{fig::msb}. The drone starts at any position, and  moves between adjacent points in the grid. It regularly  communicates its location  to a legitimate receiver, Bob; a passive eavesdropper, Eve, wishes to infer the drone's locations, and can perfectly overhear all the transmissions the drone makes. We assume the drone and Bob share just one bit of key, that is secret from Eve, and ask: what is the best use we can make of the key?

Using the one bit of shared key to protect the most significant bit (MSB) is not a good solution. As shown in Fig.~\ref{fig::msb}, the adversary can discover the fake trajectory after a few time steps since this scheme can lead to trajectories that do not adhere to the dynamics or environment constraints. At this point, it can learn the real trajectory by flipping back the MSB (we assume that the used scheme is known to everyone). Similar attacks can be made if we use a one-time pad~\cite{shannon1949communication} using the same keys over time: as time progresses, more fake trajectories can be discovered and discarded.

\begin{figure}[t]
	\centering
	\begin{subfigure}[t]{0.48\linewidth}
		\centering
		\includegraphics[width=1.3in]{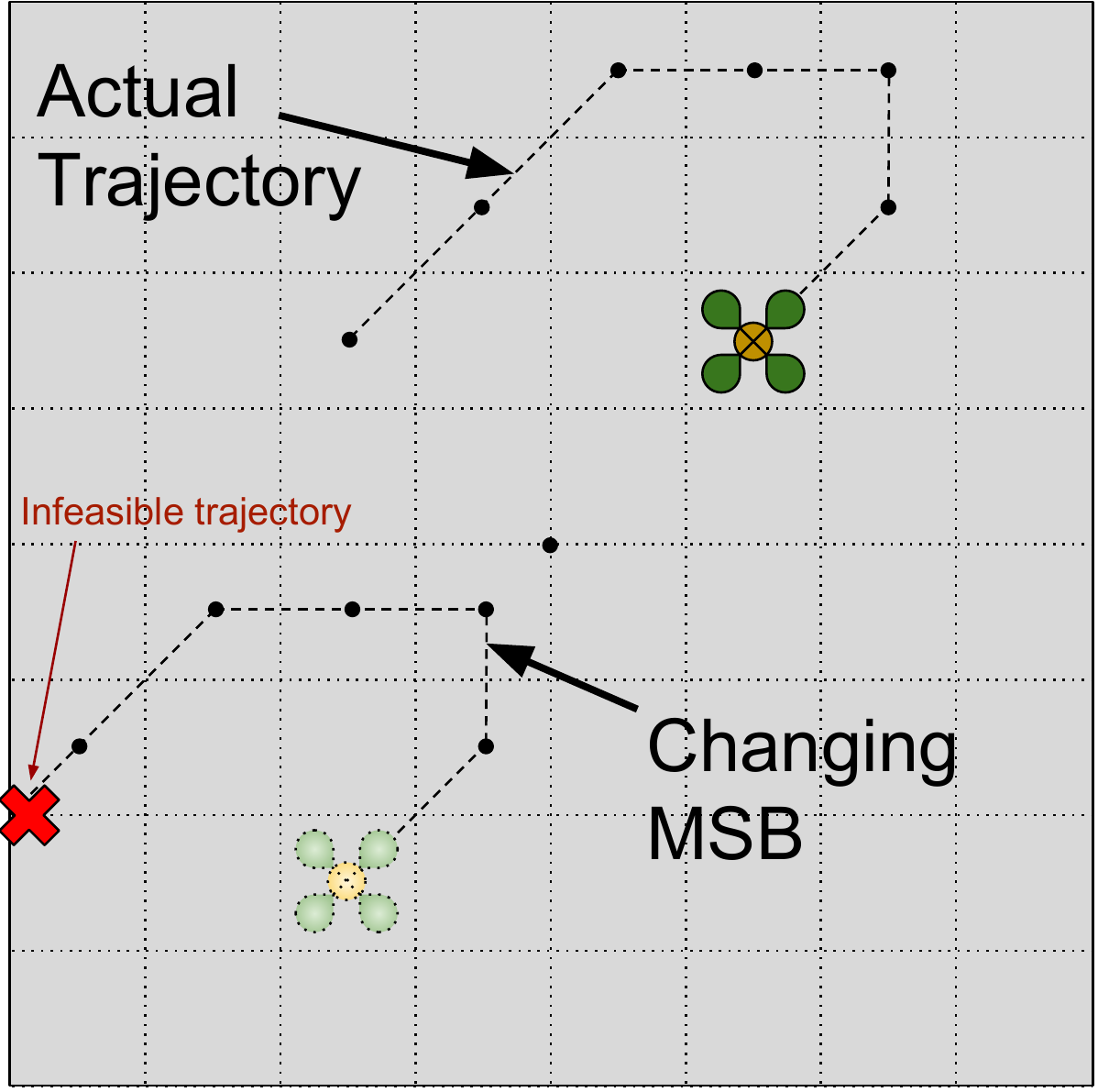}
		\caption{}
		\label{fig::msb}
	\end{subfigure}
	\begin{subfigure}[t]{0.48\linewidth}
		\centering
		\includegraphics[width=1.3in]{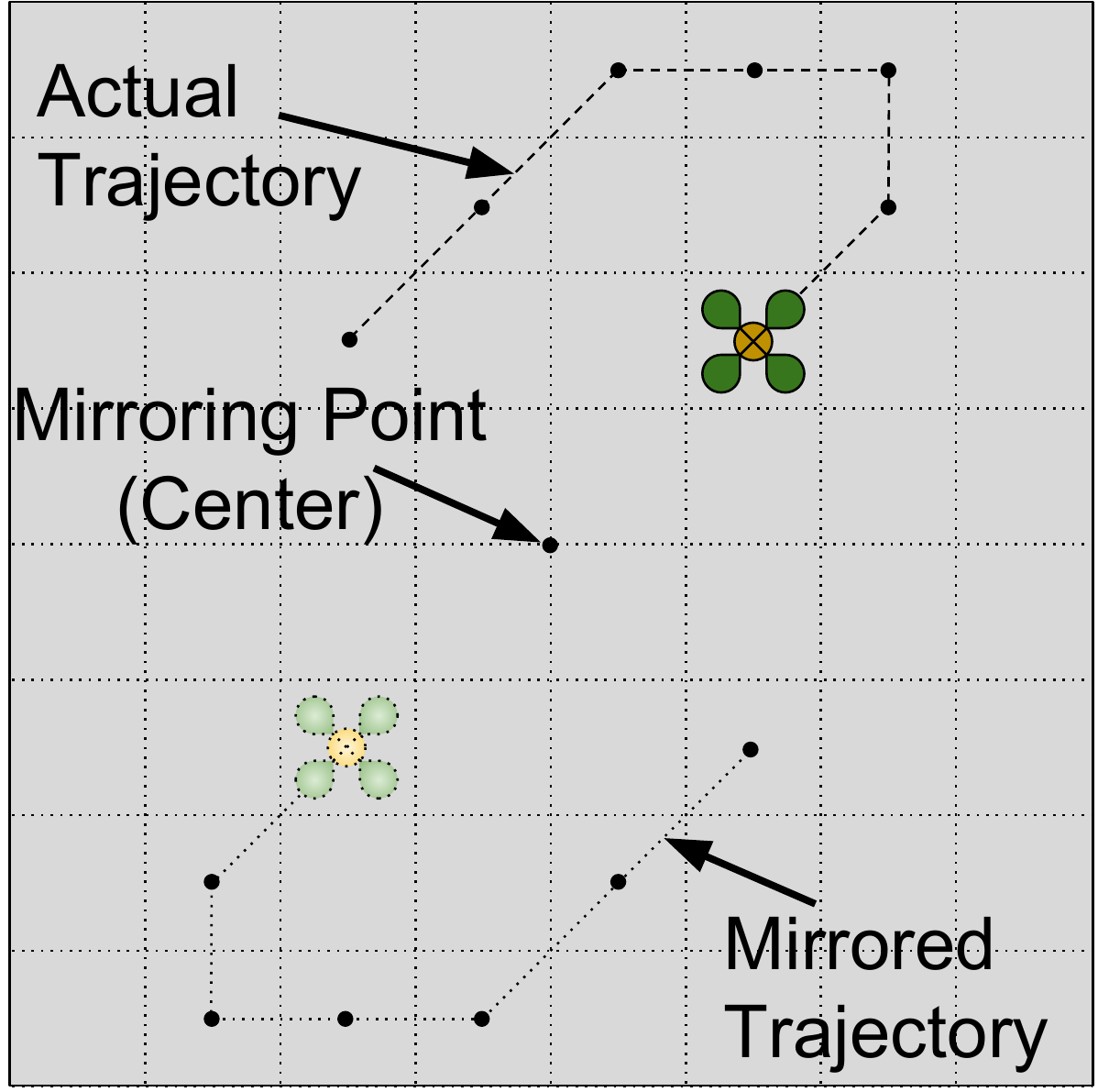}
		\caption{}
		\label{fig::mirroring}
	\end{subfigure}
	\caption{Example of drone motion: (a) protection of the most significant bit (b) mirroring based scheme.}
\end{figure}

Conventional entropy measures also fail to provide insights on how to use the  key. For instance, assume we label the $64$ squares in  Fig.~\ref{fig::msb} sequentially row per row, and consider two cases: in case I, Eve learns that  the drone is in one of the neighboring squares $\{1, 2\}$, each with probability 1/2. For case II, Eve knows that the drone is in one of the squares $\{1, 64\}$, again each with probability 1/2. Both cases are equivalent from an information security perspective since in both cases Eve's uncertainty is a set of size 2 equiprobable elements and hence its entropy is 1. However, the security risk in each situation is different. For example, if Eve aims to take a photo of the drone, in the first case she knows where to turn her camera (squares $1$ and $2$ are close by) while in the second case, she does not (squares $1$ and $64$ are far apart). 

Instead, we propose to use an Euclidean distance distortion measure: how far (in Euclidean distance) is Eve's estimate from the actual location. We then propose encoding/decoding schemes which utilize the shared key to maximize this distance. We first consider a distortion measure averaged over time and trajectories {as we formally define} later. Note that if Eve had not received any of the drone transmissions, then the best (adversarial) estimate of {the} drone's location at any given time is the center point of the confined region in Fig.~\ref{fig::msb}. Therefore, a good encryption scheme would strive to maintain Eve's estimate to be as close to the center point as possible; and we achieve the maximum possible distortion, if, after overhearing the drone's transmissions, Eve's best estimate {still remains} the center point. 

The following scheme can achieve this maximum distortion by using exactly one bit of shared secret key. When encoding, the drone either sends its actual trajectory, or a ``mirrored'' version of it, depending on the value of the secret key. The mirrored trajectory is obtained by reflecting the actual trajectory across a mirroring point in space; in this example, the mirroring point is the center point as shown in Fig.~\ref{fig::mirroring}. Since Eve does not know the value of the shared key, its best estimate of the drone's location - after receiving the drone's transmissions - would be the average location given the trajectory and its mirrored version, which is exactly the center point. 
Our  results in Section~\ref{sec:expected} extend this idea  of mirroring to dynamical systems in higher dimensional spaces, and theoretically analyze the performance in terms of average distortion for a larger variety of distributions (with certain symmetry conditions).

Next, we consider a worst-case distortion-based metric. In this case, our security metric is ``in the worst case, how far is Eve's estimate from the actual location?'' That is, the adversary's distortion may be different for different time instances and different instances of the actual trajectory, and we are interested in the minimum among these.   In Section~\ref{sec:worst} we provide encryption schemes that are suitable for maximizing this distortion metric and show that with $3$ bits of shared key per dimension (i.e., $9$ for three dimensional motion), our schemes achieve near-perfect worst case distortion.
Our main contributions are as follows:\\
	$\bullet$ We define security measures that are based on assessing the distortion: in the average sense over time and over data, and in the worst-case sense, providing such guarantees at any time and for any particular instances of data. \\
	$\bullet$ For the average distortion, we develop a mirroring based scheme which uses one bit of key and achieves the maximum possible distortion {(equivalent to Eve with no observations)} in some cases. We also discuss the cases where it is sub-optimal and analytically characterize the attained distortion.\\
	$\bullet$ For the worst case distortion, we design a scheme that uses $3$ bits {of key} per dimension and prove it achieves the {maximum possible distortion (equivalent to Eve having no observations)} when the inputs to the system are independent from the previous states.

\noindent\textbf{Related Work.}
Secure data communication where the adversary has unlimited computational power is studied from the lens on information theory, most notably by Shannon~\cite{shannon1949communication} and Wyner~\cite{wyner1975wire}. The study of secure communication from a distortion angle is relatively new and is first studied by Yamamoto~\cite{yamamoto1988rate}, where the goal is to maximize the distortion of an eavesdropper's estimate on a message. Schieler and Cuff~\cite{schieler2014rate} later showed that, in the limit of an infinite block length ($n$) code, only $\log(n)$ bits of secret keys are needed to achieve the maximum possible distortion.  Schemes for single shot communication were considered in~\cite{chiyo17} and exponential benefits for each additional bit of shared key were discussed. However, the above schemes do not directly translate to the scenarios where one has to communicate correlated temporal data like the state of a control system. 

Secure communication in control systems is studied in~\cite{Tsiamis2017StateSecrecyCF, TSIAMIS20178385,tanakaDirected, DOOREN201710090,Malik2013,Pappas16}. These works either provide distortion only at the steady state or use measures like differential privacy {(does not use keys)} and weak information theoretic security; they sometimes also assume that Eve gets different (a subset of the) information than Bob. 

%% file: SystemModel_V9MK.tex
\section{System Model}
\label{sec:sysmodel}

\noindent\textbf{Notation.} $X_t$ denote a column vector, and 
$X_{a}^{b} = [X_a^\prime \: X_{a+1}^\prime \: \cdots \: X_b^\prime]^\prime$ for $b\geq a$ and $a,b \in \mathbb{Z}$; $f_X(x)$ denotes the probability density function of a random vector $X$; for any random vector $Y$, we denote the mean vector and covariance matrix of $Y$ by $\mu_{Y}$ and $R_{Y}$ respectively, thus for example, the mean and the covariance matrix of $X_a^b$ will be denoted by $\mu_{X_a^b}$ and $R_{X_a^b}$ respectively; for a matrix $A$, $A^\prime$ denotes the transpose of $A$ and  $A^r$ denotes the $r$-th power of $A$; $[m]:= \{1,2,\ldots,m\}$ where $m \ \in \mathbb{Z^+}$.  

\noindent\textbf{System Dynamics.} We consider the linear dynamical system,
\begin{align}
X_{t+1} &= A X_{t} + B U_t + w_t,  \  \ \ 
Y_{t} = C X_{t} + v_t, \label{eq:sys_model}
\end{align}
{where $X_t \in \mathbb{R}^n$ is the state of the system at time $t$, $U_t \in \mathbb{R}^m$ is the input to the system at time $t$, $w_t \in \mathbb{R}^n$ is the process noise, $Y_t$ { are the} system observations, and $v_t \in \mathbb{R}^p$ is the observation noise.}
{Let} $X = X_1^T$, $U = {U_{1}^{T-1}}$ and ${w = w_1^{T-1}}$.
Based on the {initial $X_1$ and target $X_T$ states}, the controller computes {$U_1^{T-1}$ which} moves the system from $X_1$ to $X_T$.

\noindent\textbf{Communication and Adversary Models.} At each time instance the system transmits information about its state to a legitimate receiver, which is referred to as Bob, via a noiseless link. This situation occurs for example when Bob is remotely monitoring the execution of the system as in {Supervisory Control And Data Acquisition (SCADA)} system. A malicious receiver, Eve, is assumed to eavesdrop on the communication between the system and Bob and is able to receive all transmitted signals. Eve is assumed to be passive: she does not actively communicate but is interested in learning the system's states from $t=1$ to $T$.
We assume that the system and Bob have a shared { $k$-bit} key $K$ which they use to encode/decode the transmitted messages.

\noindent\textbf{Inputs and States Random Process Model.} We assume that both Eve and Bob are aware of the system model, the matrices $A, B, C$ and the statistics of noises. From the perspective of Eve, the input and output sequences have random distributions which depend on $A, B, C$ and the statistics of the noise. In addition to the process noise $w$, the joint distribution $f(X,U,w)$ depends on the initial and target states and the control law of the system. So, even in noiseless systems, $X$ and $U$ possess inherent randomness from Eve's perspective due to her lack of knowledge about the control law and the initial and target states. In general, the control inputs $U$ can be dependent on the system states $X$. However, knowing the marginal {distribution} of {$U$ in noiseless systems} can specify the marginal {distribution} of $X$. This follows by noting that $X_2^T = Q U + \tilde{Q} W + \hat{Q} X_1$, where $Q$ and $\tilde{Q}$ are lower triangular {block matrices with the $(i,j)$th block submatrices, $i \geq j$, being $A^{i-j}B$ and $ A^{i-j}$ repsectively,} and $\hat{Q} = [A^\prime \: ({A^{2}})^\prime \: \cdots \: ({A^{T-1}})^\prime ]^\prime$. This implies that for noiseless systems, the marginal distribution of $U$ would imply the marginal distribution of $X_2^T$ for a given initial state $X_1$ and thus the marginal distribution of $X$. {For a given $X_1$}, the mean vector and covariance matrix of $X_2^T$ become $\mathbb{E}(X_2^T) = \mu_{X_2^T} = Q \mu_U + \hat{Q} \mu_{X_1}$ and $R_{X_2^T} = Q R_U Q^T$.

\noindent \textbf{Encoding Model.} The system transmits a packet $Z_t$ at each time step $t$.
%
%
%
%
The $t$-th transmitted packet can be a function of all previous observations and the shared keys, thus,
$Z_t \coloneqq \mathcal{E}_t (Y_1^{t}, K)$, where  $\mathcal{E}_t$ is the encoding function used at time $t$. We will denote $Z_1^T$ by $Z$.\\
\noindent\textbf{Bob/Eve Models of Decoding.}
Bob {noiselessly receives the} transmitted packets from the {system, and decodes them using the shared key. Then, using the decoded information, it generates an estimate of the state transmissions of the system at times $t=1,\cdots,T$.} 
We require Bob to decode losslessly (i.e., with zero distortion). Formally, $H(X_t|Z_1^t, K)=0, \ \forall t \in [T]$, where $H$ is the Shannon entropy~\cite{shannon1949communication}.

Similarly, Eve also {receives} all transmissions from the system. {However,} unlike Bob, she does not have the key $K$. 
Therefore, Eve's estimate of $X_t$ is $\hat{X}_{t} \coloneqq \Set{\phi}_t \left(Z_1^T\right),  t \in [T]$,
where $\phi_t$ is the decoding function used by Eve at time $t$.

\noindent\textbf{Distortion Metrics.} We consider a distortion-based security metric which captures how far (in Euclidean distance) an estimate is from the actual value. More formally, for a given time instance $t$ and a transmitted codeword $Z_1^T$, we {define}
\begin{align}
D(t,Z_1^T) & \coloneqq \mathbb{E}_{X_t|Z_1^T} \left\Vert X_t - \hat{X}_t \right\Vert_2^2 \stackrel{(a)}{=} \text{tr}\left( R_{X_t|Z_1^T} \right), \label{DE_single}
\end{align} 
\noindent where \eqref{DE_single} captures the distortion incurred by Eve's estimate of $X_t$.
Equality in (a) follows because the best (minimizing) estimates of Eve at time $t$ are, $
 \hat{X}_t  = \Set{\phi}_t \left(Z_1^T\right) = \mathbb{E} \left[X_t|Z_1^T \right].$ 
 {This implies that Eve's state estimation is the optimal one given the observations $Z_1^T$. In general, this state estimate is dependent on the time instance. In other words, unless it happens to be the optimal estimate, making a constant estimation of the state hoping that it matches the actual state at some time will lead to high distortion values.}
Because Bob is required to successfully decode - for a given realization of the key, the encoding function can only map one $X_t$ and that key realization to each value of $Z_1^T$. Therefore Eve realizes {that only trajectories from a {particular subset} can be the true trajectory} for a given $Z_1^T$: those are the ones which correspond to each key realization. The expectation in \eqref{DE_single} is in fact taken over the randomness in the key {taking into account} {posterior} probabilities {given} $Z_1^T$. 
{If Eve does not have observations, the expectation is taken} over $X_t$ with prior distribution and will get $D(t,Z_1^T) = \text{tr}(R_{X_t})$.

As $D(t,Z_1^T)$ is a function of time $t$ and the transmitted sequence $Z_1^T$, we consider two overall distortion metrics: the average case distortion (denoted by $D_E$) where we take the expectation over all possible $Z_1^T$ averaged out over time; and the worst case distortion (denoted by $D_W$) where we take the minimum over all possible $Z_1^T$ and time instances.
{
\begin{align}
&\begin{array}{ll}
\text{Average} \\ \text{Distortion}
\end{array} - \:\:
D_E \coloneqq \mathbb{E}_{Z_1^T} \left[ \frac{1}{T} \sum\limits_{t=1}^T D(t, Z_1^T) \right]  \label{eq:average_case}\\ 
&\begin{array}{ll}
\text{Worst Case} \\ \text{Distortion}
\end{array} - \:\: 
D_W \coloneqq \min\limits_{Z_1^T} \left[ \min\limits_{t\in [T]} D(t, Z_1^T) \right] \label{eq:worst_case} .
\end{align} 
}

{Note that} $D_W$ can be defined even when there is no prior distribution on  $X_1^T$. However, {to provide} a baseline comparison {with} the case when the adversary has no observations, we assume that $X_1^T$ always have a known prior distribution.

\noindent\textbf{Design Goals.} Our goal is to choose the encoding function, $\mathcal{E}_t$, so that Bob can decode loselessly while the distortion is maximized for Eve's estimate. In addition, we seek to achieve this with the minimum amount of shared key $K$. In absence of any observation by Eve, these distortions will be, $D_E^{\text{max}} =  \frac{1}{T} {\sum_{t=1}^T} \text{tr}(R_{X_t})$ and $D_W^{\text{max}} = {\min_{t \in [T]}} \text{tr}(R_{X_t})$.
These {provide} upper bounds as,
\begin{align}
D_E & \!=\!\frac{1}{T} \mathbb{E}_{Z_1^T}\! \sum\limits_{t=1}^{T}\!\text{tr}(R_{X_t|Z_1^T})\!\stackrel{(a)}{\leq}\!\frac{1}{T}\!\sum\limits_{t=1}^{T}\!\text{tr}(R_{X_t})\!=\!D_E^{\text{max}},\label{eq::UpperBound} \\
D_W  &= \min\limits_{Z_1^T}  \min\limits_{t\in [T]} \text{tr}(R_{X_t|Z_1^T})  \leq   \min\limits_{t\in [T]}\mathbb{E}_{Z_1^T} \left[ \text{tr}(R_{X_t|Z_1^T})\right] \nonumber \\
 &\stackrel{(b)}{\leq} \min\limits_{t\in [T]} \text{tr}(R_{X_t}) = D_W^{\text{max}} \label{eq::UpperBoundDW} ,  
 \end{align}
\noindent where (a) and (b) follow by noting that the trace of the conditional covariance matrix is a quadratic (convex) function in $Z_1^T$ and therefore we can use Jensen's inequality.

%% file: AverageDist_V10MK.tex
\section{Optimizing The Average Distortion $D_{E}$}
\label{sec:expected}
In this section, we assume that the control system in \eqref{eq:sys_model} is noise free, that is $v_t = w_t = 0$. Although our results can be extended to an arbitrary observable pair (A,C) in~\eqref{eq:sys_model}, to simplify the exposition we assume the state can be directly measured (C = I).
We now discuss our proposed scheme that uses one bit of shared key and show how the achieved distortion compares to the upper bound in~\eqref{eq::UpperBound}. 
{ As we show later (Corollary~\ref{cor::SymmetricDist})}, this scheme is optimal when the prior distribution on the state have a {point of symmetry}.

\noindent\textbf{Mirroring Scheme.} Let $\tilde{X}_t$ be the state vector $X_t$, mirrored across a affine subspace $\mathcal{V}_t = \{ x \in \mathbb{R}^{n} \: | \: S_t x = b_t \}$,
where $S_t \in \mathbb{R}^{s_t \text{x} n}$ and $b_t \in \mathbb{R}^{s_t} $ This scheme works as follows:
\begin{equation}
Z_t = (1-K) X_t + K \tilde{X}_t, \ \forall t \in [T],
\end{equation}
where $K$ is the shared bit. Since every affine subspace can be written in terms of orthogonal vectors, we assume that $S_tS_t^\prime = I$. 
It is easy to show that the mirrored point $\tilde{X}_t$ is  $(I - 2S_t^\prime S_t)X_t + 2S_t^\prime b_t$ and thus the encoding/decoding complexity of our scheme is $O(n^2)$. 

\textbf{Example.} Consider $X_t \in \mathbb{R}^2$ where $S_t = \frac{1}{\sqrt{2}}[-1 \:\: 1]$ and $b_t = 0$. Then $\tilde{X}_t$ corresponds to reflecting across a line that passes through the origin with a $45^o$ angle. 

The performance of our scheme is as follows.
{
\begin{theorem}(Proof in Appendix V-A)
\label{thm::MirroringPerformanceCausal}
The mirroring scheme with matrices $S_t$ and $b_t$ allows Bob to perfectly estimate $X_t$, and the distortion for Eve is,		
\begin{equation}
\label{eq::DEE_rough_causal}
D_E = \frac{1}{T} \sum\limits_{t=1}^T  \mathbb{E}_{X} \left[ \frac{2 f_{X}(\tilde{{X}})}{f_{X}({X})+f_{X}(\tilde{{X}})} \|S_t{X_t} - b_t\|^2 \right],
\end{equation}
\noindent where $\tilde{X}: = [{\tilde{X}_1}^\prime \:\: {\tilde{X}_2}^\prime \:\: \cdots \:\: {\tilde{X}_T}^\prime]^\prime$ is the mirrored version of $X$. 
\end{theorem}

Assuming that $f_{X}(x)$ is known, then Theorem~\ref{thm::MirroringPerformanceCausal} provides a closed-form characterization of the achieved average distortion for any mirroring scheme with matrices $S_t$ and $b_t$. Moreover, under some symmetry conditions on $f_{X}(x)$, the expression in \eqref{eq::DEE_rough_causal} simplifies and gives insights on the maximum achievable distortion. This is shown in Corollary~\ref{cor::simplified_mirror}.
\begin{corollary}(Proof in Appendix V-A)
	\label{cor::simplified_mirror}
	If the mirroring scheme matrices $S_t$ and $b_t$ in Theorem~\ref{thm::MirroringPerformanceCausal} are selected such that
	$f_{X}(X) = f_{X}(\tilde{X}), \forall X \in \mathbb{R}^{n T}$, then \eqref{eq::DEE_rough_causal} becomes,
	\begin{equation}
	\label{eq::DEE_simple_causal}
	D_E = \frac{1}{T} \sum\limits_{i=1}^T\!\!\text{tr}\left( S_t R_{X_t} S_t^\prime + (b_t\!-\!S_t \mu_{X_t})( b_t\!-\!S_t \mu_{X_t})^\prime \right)
	\end{equation}
\end{corollary}
%
Note that $f_{X}(X) = f_{X}(\tilde{X})$ implies $b_t = S_t \mu_{X_t}$. We can interpret~\eqref{eq::DEE_simple_causal} as follows. 
Assuming that $f_{X}(X) = f_{X}(\tilde{X})$ is met, then the distortion becomes $D_E = \frac{1}{T} \sum\limits_{i=1}^{T} \text{tr}(S_t R_{X_t}S_t^\prime)$. The achieved distortion therefore depends on the choice of $S_t$: if $S_t = I$ then the maximum distortion can be achieved by our mirroring scheme. However, such a choice of $S_t$ may not be able to ensure that $f_{X}(X) = f_{X}(\tilde{X})$ is met, as we will see in some of the following examples.
One case for which $S_t = I$ satisfies $f_{X}(X) = f_{X}(\tilde{X})$ and allows maximum distortion is when $X$ is symmetrically distributed around a point. We show this in the next corollary.}
	
	\begin{corollary}
		\label{cor::SymmetricDist}
		For a random vector $X \in \mathbb{R}^{Tn}$, if there exists a point $v \in \mathbb{R}^{Tn}$ for which $f_{X}(X) = f_{X}(2v - X)$, $\forall X \in \mathbb{R}^{Tn}$, then $D_E = \frac{1}{T} \text{tr}(R_X) = \frac{1}{T}\sum\limits_{t=1}^{T} \text{tr}(R_{X_t})$.
	\end{corollary}
	\begin{proof}
Since $X$ and $2v-X$ have the same distribution, they will have the same mean. This implies that {$v = \mu_X$}. We then use the following mirroring scheme: $S_t = I$, $b_t = \mu_{X_t}$ for $t \in [1:T]$. With this, we get
$\tilde{X_t} = 2 \mu_{X_t} - X_t$, and thus $\tilde{X} = 2 \mu_{X} - X$ where $\tilde{X}: = [{\tilde{X}_1}^\prime \:\: {\tilde{X}_2}^\prime \:\: \cdots \:\: {\tilde{X}_T}^\prime]^\prime$ and $\mu_X: = [{\mu_{X_1}}^\prime \:\: {\mu_{X_2}}^\prime \:\: \cdots \:\: {\mu_{X_T}}^\prime]^\prime$. This implies, $f_X(X) = f_X(\tilde{X}),  \ \forall X \in \mathbb{R}^{n T}$. Therefore the distortion is $D_E^{\max}$.
\end{proof}

We now illustrate our results for few examples.\\
%
%
\noindent\textbf{Example 1.} 
Assume $U$ is distributed as Gaussian with mean $\mu_U$ and covariance matrix $R_U$. Then {for a zero initial state, $X_2^T$} is also Gaussian distributed with mean ${\mu_{X_2^T}} = Q \mu_U$ and covariance ${R_{X_2^T}} = Q R_U Q^T$, as we assume the noise to be zero. A Gaussian random vector satisfies the conditions in Corollary \ref{cor::SymmetricDist}, and therefore we can get maximum distortion by setting $b_t = \mu_{X_t}$ and $S_t = I$.
\iftrue
	
The next example is based on a Markov-based model for the dynamical system and uses the following lemma.
	
\begin{lemma}
\label{lem::MarkovModel}
Consider the random vectors $X_t$ where the following conditions hold: 1) $f_{X_1}(x_1) = f_{X_1}(2\mu_1 - x_1)$ and 2) $f_{X_t|X_{t-1}}(x_t | x_{t-1}) = f_{X_t|X_{t-1}}(2 \mu_t - x_t | 2\mu_{t-1} - x_{t-1})$. Then for this case, $f_{X}(X) = f_{X}(2 \mu - X)$, where $\mu = [{\mu_1}^\prime \:\: {\mu_2}^\prime \:\: \cdots \:\: {\mu_T}^\prime]^\prime$. Therefore, by virtue of Corollary \ref{cor::SymmetricDist}, mirroring schemes can achieve the maximum distortion.
\end{lemma}
	
\noindent\textbf{Example 2.} Consider {the following random walk} mobility model. Let $a \in \mathbb{N}^+$, and $X_t$ be its location at time $t$, then,	
\begin{align*}
X_1 &\sim \text{Uni}([-a:a]) \\
X_{t}|X_{t-1} &\sim \text{Uni}([-a:a] \cap \{X_{t-1} -1, X_{t-1}, X_{t-1}+1 \}).
\end{align*}	
One can see that these distributions satisfy the conditions in Lemma \ref{lem::MarkovModel}. Therefore, one can set $b_t = \mu_t = 0$ and $S_t = 1$, which will achieve maximum distortion of $D_E$.
\fi

{
\noindent\textbf{Example 3.} Here we provide a numerical example which shows how our mirroring scheme performs for situations where we do not have an analytical handle on the state distributions. We assume the quadrotor dynamical system provided in~\cite{kumar2012opportunities}. The quadrotor moves in a 3-dimensional cubed space with a width, length and height of 2 meters, where the origin is the center point of the space. The quadrotor starts its trajectory from an initial point $(-1,y_1,z_1)$ and finishes its trajectory at a target point $(1,y_T,z_T)$ after $T$ time steps, where {the points $y_1,z_1,y_T,z_T$ are picked uniformly at random in $[-1,1]^4$}. We assume that $T=10$ time steps, and that the continuous model in~\cite{kumar2012opportunities} is discretized with a sample time of $T_s = 0.5$ seconds. We assume that the quadrotor encodes and transmits only the states which contain the location information {(first three elements of the state vector $X_t$)}. The quadrotor is equipped with an LQR controller which designs the input sequence $U_1^{T-1}$ {which minimizes}
$\left\| U \right\|^2 + 10 \left\| X_2^{T-1} \right\|^2$ while ensuring that $X_T$ is equal to the target state.
%
We perform numerical simulation of the aforementioned setup: we run {$2$ millions} iterations, where in each iteration a new initial and target points are picked, and the resultant trajectory is recorded. Based on the recorded data, we consider different mirroring schemes and numerically evaluate the attained distortion. To facilitate numerical evaluations, the simulation space is gridded into bins with $0.2$ meters of separation, and the location of the drone at each trajectory is approximated to the nearest bin. 

\begin{figure}
 \centering
 \includegraphics[width=2.5in]{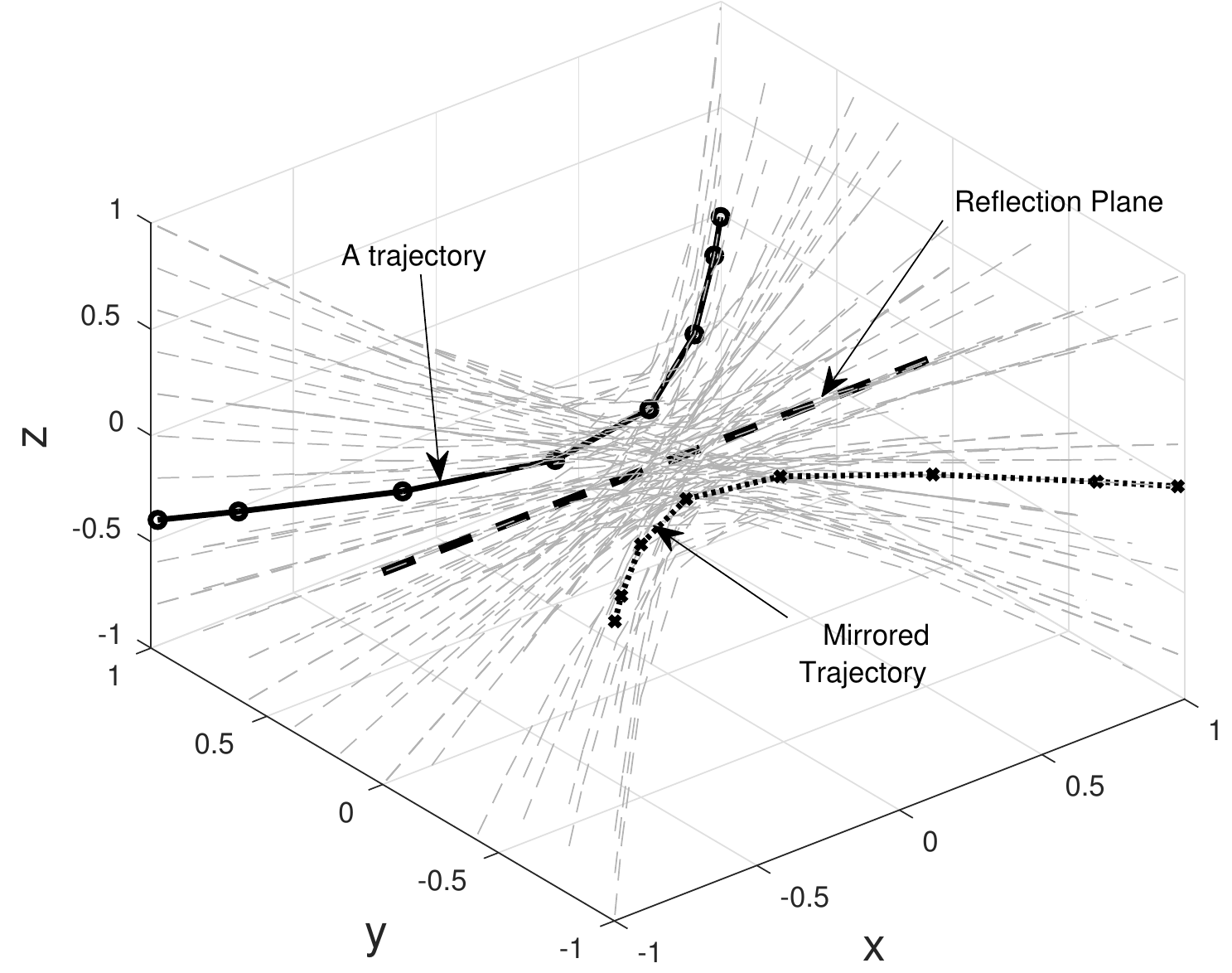}
 \caption{An illustration of some trajectories. The reflection plane is shown as a dashed-black line. One trajectory (solid-black) is shown along with its mirrored image (dotted-black).}
 \label{fig::simulationMirroring}
\end{figure}

Figure \ref{fig::simulationMirroring} shows some of the drone trajectories obtained from our numerical simulation. It is clear that not all trajectories are equiprobable, and therefore the distribution of $X_t$ is not uniform across all bins in space. However, the computation of $\mathbb{E} X_t$ shows {the expected value of the position} to be the origin. Moreover, since the motion of the drone is mainly progressive in the positive x-axis direction, reflection across the origin results in mirrored trajectories that are progressing in the opposite direction, and therefore are identified to be fake automatically. Therefore, mirroring across a point here is useless: the numerically computed distortion for this scheme is equal to zero.

Next we consider mirroring across the reflection plane shown in Figure~\ref{fig::simulationMirroring}, where $b_t = 0$ and $S_t = \left[ \begin{matrix}0 & 1 & 0 \\ 0 & 0 & 1 \end{matrix}\right]$. As can be seen from the figure, the reflection plane is indeed an axis of symmetry for the distribution of the drones trajectories, and therefore is expected to provide high distortion values. We numerically evaluate the attained distortion using the scheme by using equation \eqref{eq::DEE_rough_causal}, which evaluates to $D_E = 0.3971$. This is slightly less than $D_E^{\max} = 0.3979$.

%% file: WorstDist_V12MK.tex
\section{Optimizing The Worst Case Distortion $D_W$}
\label{sec:worst}


\begin{figure}[h]
	\centering
		\includegraphics[width=0.8\linewidth]{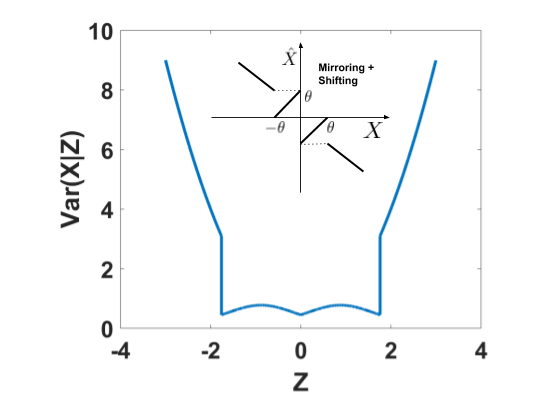}
		\caption{$\text{Var}(X|Z)$ Vs Z for mirroring+shift based scheme with $\theta_1=1.76$; $D_W = 0.4477$.}
		\label{fig:plot_distortion_mirror_shift}
\end{figure}

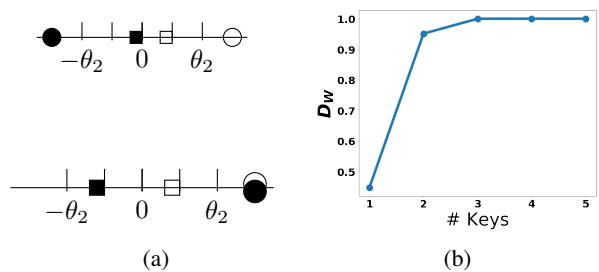
\begin{figure}[ht]
	\centering
	\begin{subfigure}[t]{0.45\linewidth}
		\begin{tikzpicture}
		\begin{scope}[scale=0.4]
		\draw (-3.5,0) -- (3.5,0);
		\draw (0,-0.1) node [below] {$0$};
		\draw (2,-0.1) node [below] {$\theta_2$};
		\draw (-2,-0.1) node [below] {$-\theta_2$};
		\draw (-2,-0.1) -- (-2,0.5);
		\draw (2,-0.1) -- (2,0.5);
		\draw (0,-0.1) -- (0,0.5);
		\draw (0,-0.1) -- (0,0.5);
		\draw (1,-0.1) -- (1,0.5);
		\draw (-1,-0.1) -- (-1,0.5);
		\draw (0.6,-0.2) rectangle (1, 0.2);
		\draw[fill=black] (-0.4,-0.2) rectangle (0, 0.2);
		\draw (3,0) circle(3mm);
		\draw[fill=black] (-3,0) circle(3mm);
		\end{scope}
		\begin{scope}[scale=0.5, shift={(0,-4)}]
		\draw (-3.5,0) -- (3.5,0);
		\draw (0,-0.1) node [below] {$0$};
		\draw (2,-0.1) node [below] {$\theta_2$};
		\draw (-2,-0.1) node [below] {$-\theta_2$};
		\draw (-2,-0.1) -- (-2,0.5);
		\draw (2,-0.1) -- (2,0.5);
		\draw (0,-0.1) -- (0,0.5);
		\draw (0,-0.1) -- (0,0.5);
		\draw (1,-0.1) -- (1,0.5);
		\draw (-1,-0.1) -- (-1,0.5);
		\draw (0.6,-0.2) rectangle (1, 0.2);
		\draw[fill=black] (-1.4,-0.2) rectangle (-1, 0.2);
		\draw (3,0.1) circle(3mm);
		\draw[fill=black] (3,-0.1) circle(3mm);
		\end{scope}
		\end{tikzpicture}		
		\caption{}
		\label{fig:scheme}
	\end{subfigure}
	\begin{subfigure}[t]{0.45\linewidth}
		\centering
		\includegraphics[width=1.5in]{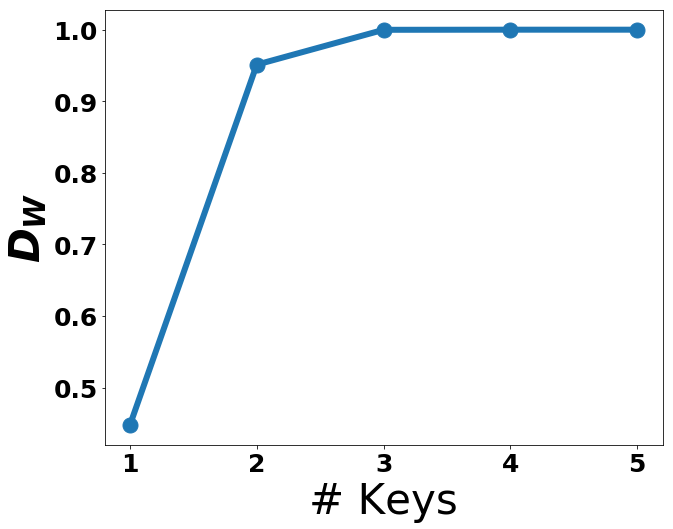}
		\caption{}
		\label{fig:dist_vs_keys}
	\end{subfigure}
	\caption{(a) Transparent shapes represent true values and solid shapes represent their respective mapping when two bit key is $11$ and $10$ respectively. (b) $D_W$, as a function of number of keys for optimal choice of $\theta_k$.}		
\end{figure}

The expected distortion metric might not be well-suited for some applications (for example if an adversary wants to shoot a drone). In this case, the adversary's estimate needs to be far from the actual state \textit{at all} time instances. Therefore, a more appropriate metric would be to consider the worst case distortion for the adversary. Consider for example the scheme in Fig.~\ref{fig::mirroring}. Here, the adversary's estimate is always the center point and the maximum expected distortion is achieved. However, when the drone is very close to the center, its mirror image will also be close to the center.  At this particular time instance, the adversary's distortion will be very small and thus the adversary will essentially know the position.

In this section, we present an encryption scheme that attempts to maximize 
the worst  case distortion for Eve. 
{The scheme}
obfuscates the initial state 
{such that,}
even if {Eve} optimally uses her {observations and} knowledge about the dynamics, her best estimate {attains maximal} distortion. We start by studying 
{the case of single shot transmission} (Theorems~\ref{thm:gaussian_scalar} and~\ref{thm:gaussian_vector}), which form the basis for maximizing the worst case distortion of a trajectory {(Theorem~\ref{thm:traj})}. 


\subsection{Building Step: Scalar Case}
Consider the case where the system wants to communicate a single scalar random variable $X$ to Bob by transmitting $Z$. The worst case distortion  $D_W$ for Eve will be $D_W = \min_{Z} \text{Var}(X|Z)$. Note that if Eve does not overhear $Z$, Eve {uses} the minimum mean square {error} estimate ({i.e.,} the mean value) as her estimate,
and thus experiences a worst case distortion equal to the variance of $X$. 

We first assume that $X \sim \Set{N}(0,1)$, and thus, the worst case distortion can not be larger than $1$ by~\eqref{eq::UpperBoundDW}. We next
develop our scheme progressively, from simple to more sophisticated steps.
We will also use the following lemma.
\begin{lemma}
	\label{lem:var}
	The variance of two real numbers $a$ and $b$ with probabilities $p_a$ and $p_b$ is given by $p_a p_b(a-b)^2$.
\end{lemma}

\noindent\textbf{Mirroring or Shifting.} Reflecting around the origin ({as proposed in} Section~\ref{sec:expected}) {is not suitable for maximizing $D_W$:} indeed, 
{using Lemma~\ref{lem:var}, $\text{Var}(X|Z)$ evaluates to $0$ when $Z=0$ and attains limited distortion for small values of $Z$.}
{Another scheme consists of ``shifting'' $X$ by a constant $\theta$}
whenever the shared key bit is one. 
{Differently, this scheme admits  $\text{Var}(X|Z)$ which decreases fast for large values of $Z$.}

\noindent\textbf{Shifting+Mirroring.} We here combine {both schemes} in order to achieve a good performance for both small and large values of  $X$. We  start from the case where we have $k=1$ bit of key and then go to the case $k\geq 1$. \\
$\bullet \quad { k=1}.$ We select a $\theta_1\in \mathbb{R}$ that determines a window size ($\theta_1 $ is public and known by Eve). The encoding function is
{\begin{align*}
Z & = \mathcal{E} (X, K) = \left\{\begin{array}{cl} (1-2K)X & \text{if } K = 0 \ \text{or} \ |X| > \theta_1 \\ X + \theta_1 & \text{if } K = 1, \ -\theta_1 \leq X < 0 \\ X - \theta_1 & \text{if } K = 1, \ 0 \leq X < \theta_1\end{array} \right.
\end{align*} }
We note that there is one particular value of $X$,  $X = \theta_1$, which we do not transmit. Since this is of zero probability measure, it can be safely ignored.
Given $Z$, there are two possibilities for $X$: for $|Z| > \theta_1, X \in \{Z,-Z\}$; for $-\theta_1 \leq Z < 0, X \in \{Z,Z+\theta_1\}$; for $0 \leq Z < \theta_1, X \in \{Z, Z-\theta_1\}$.  
Using the fact that $X \sim \Set{N}(0,1)$, we can calculate the posterior probabilities $Pr(X|Z)$ and use Lemma~\ref{lem:var}  to compute $\text{Var}{(X|Z)}$. Fig.~\ref{fig:plot_distortion_mirror_shift}
plots $\text{Var}(X|Z)$ for {$\theta_1=1.76$}. The worst case distortion in this case becomes $0.4477$, which is the best we can hope for if we have only one bit of shared key. {This follows because for any mapping from $X$ to $Z$, a transmitted symbol $Z$ can have at most two pre-images (as Bob needs to reliably decode with one bit of key), and if one of these is $X=0$, then no matter what the second one is, the distortion corresponding to  $Z$ will be at most $0.4477$. Equality occurs when the second pre-image of  $Z$ is $\pm 1.76$. Note that our scheme also maps {$0$} to { $-1.76$} (for $\theta_1 = 1.76$}).}  \\
\noindent $\bullet \quad k\geq 1$. For $K \in \{0,1\}^k$,
we use the following encoding:
\begin{align*}
Z & = \mathcal{E} (X, K) = \left\{ \begin{array}{lr}
 X \:\:\:\:\:\:\:\:\:\:\:\quad\quad\quad\quad  K < 2^{k-1}, |X| > \theta_k  \\ 
 -X \:\:\:\:\:\:\:\quad\quad\quad\quad  K \geq 2^{k-1}, |X| > \theta_k  \\
\!X\!+\!K \frac{2\theta_k}{2^{k}} \text{ mod } [-\!\!\theta_k, \theta_k) \quad\quad\quad\:\:\:  \text{o.w. },\end{array} \right.
\end{align*}
\noindent where the optimal value of  the constant $\theta_k$ depends on the number $k$ of keys we have, $K$ is the decimal equivalent of a binary string of length $k$, and $r \text{ mod } [a,b) = r - i(b-a) $ is such that $i$ is an integer and $r - i(b-a) \in [a,b)$ for $r,a,b \in \mathbb{R}$. Intuitively, if $|X|> \theta_k$ then for half of the keys, we reflect across origin and for the other half we do nothing; if $|X| < \theta_k$, we divide this window of size $2\theta_k$ into ${2^k}$ equal size windows and shift a point from one window to another by jumping $K$ (in decimal) windows. An example for $k=2$  is shown in Fig.~\ref{fig:scheme} for the key values $K=11$ and $K=10$. Fig.~\ref{fig:dist_vs_keys}  plots $D_W$ as a function of the number of keys $k$. Using $k=3$ and $\theta_3=4.84$ we achieve $D_W=0.9998$ which is very close to $1$, the best we could hope for.

\begin{theorem}
	\label{thm:gaussian_scalar}
	A Gaussian random variable with mean $\mu$ and variance $\sigma^2$ can be near perfectly {($\sim 0.9998$ times the perfect distortion)}
	distorted in the worst case settings by just using three bits of shared keys.
\end{theorem}
\begin{proof}
Generate the random variable  $V \sim \Set{N}(0,1)$ as $V = (X - \mu)/\sigma$ and encrypt it using $k=3$ key bits and the previously described scheme. For  $c = 0.9998$ we have $D_w= \min\limits_Z \text{Var}(X|Z)  = \min\limits_Z \text{Var}(\sigma V + \mu |Z)  = c\sigma^2$.
\end{proof}
\textbf{Remark:} We optimized the parameter $\theta_k$ of our scheme assuming Gaussian distribution. For other distributions, the optimal choice of $\theta_k$ and the corresponding worst case distortion would be different.
\subsection{Vector Case and Time Series}
\begin{theorem}(Proof in Appendix V-B)
	\label{thm:gaussian_vector}
	For a Gaussian random vector $X \in \mathbb{R}^n$ with mean $\mu$ and a diagonal covariance matrix $\Sigma$ we can achieve $D_W$ within $0.9998$ of the optimal by using $3n$ bits of shared keys.
\end{theorem} 
This theorem uses our 3-bit encryption for each element {in} the vector. Assume now that this vector captures the probability distribution of the initial state of dynamical system; by encrypting this state we can guarantee the following. 
\begin{theorem}(Complete Proof in Appendix V-C)
  \label{thm:traj}
  Using $3n$ bits of shared keys we can achieve $D_W$ within $0.9998$ of the optimal for the dynamical systems~\eqref{eq:sys_model} with $C = I$, $v_t = 0$,  singular values of $A$ more than $1$, and initial state $X_1 \sim \mathcal{N}(\mu, \Sigma)$, where $\Sigma$ is diagonal covariance matrix, and $U_t$ and $w_t$ are independent of $X_t$.
\end{theorem}
{\textbf{Remark}: {Although the independence assumption on the inputs is rather restrictive, the result}
serves as a stepping stone {towards understanding general cases}}.

\begin{proof}
 The system transmits $ Z_1  = f(Y_1, K) = f(X_1,K)$ where $f$ is the encoding in Theorem~\ref{thm:gaussian_vector}, and
	\begin{align*}
	Z_{t\!+\!1}\!=\!A Z_t\!+\!(Y_{t\!+\!1}\!-\!A Y_{t}) = A Z_t\!+\!B U_t\!+\!w_t,  \forall t \in [T-1]. 
	\end{align*}
 Bob can  decode $X_1$ using  $Z_1$ and $K$.  Then:
	\begin{align*}
	\hat{X}_{t+1} & = Z_{t+1} - A Z_t + A \hat{X_t}  = B U_t + w_t + A\hat{X_t}\\
	& = A X_t  + BU_t + w_t = X_{t+1} , \ \ \ \forall t \in[T-1].
	\end{align*}
	Eve's distortion is calculated  in Appendix V-C. \end{proof}

\noindent\textbf{Complexity:} $\mathcal{O}(n^2)$ per time for both encoding \& decoding.

%% file: Appendices_V3GA.tex
\section{Appendices}

\subsection{Proof of Theorem~\ref{thm::MirroringPerformanceCausal} and Corollary~\ref{cor::simplified_mirror}}
\label{app::MirroringPerformanceCausal}
%

We start by computing $R_{X_t|Z_1^T}$. Note that given a sequence of transmitted symbol $Z_1^T$ there are two possible values of sequence of message symbols $X_1^T$ which are $X_1^T = Z_1^T$ and $X_1^T = \tilde{Z}_1^T$, where $\tilde{Z_t}$ is the image of $Z_t$ across the affine subspace given by $S_tx = b_t$. 

{With this, the posterior probability of $X_t=Z_t$ given $Z_1^T$ i.e., $Pr (X_t = Z_t |Z_1^T)$ will be equal to {$Pr (X_1^T = Z_1^T|Z_1^T) \coloneqq p_Z$. We note that $p_Z = \frac{f(Z)}{f(Z) + f(\tilde{Z})}$}, where $\tilde{Z}: = [{\tilde{Z}_1}^\prime \:\: {\tilde{Z}_2}^\prime \:\: \cdots \:\: {\tilde{Z}_T}^\prime]^\prime$.} Then $\mathbb{E}(X_t|Z_1^T)$,

\begin{align*}
&= p_Z {Z_t} + (1-p_Z)(\tilde{{Z_t}}) = {Z_t} + 2(1-p_Z)S_t^\prime\left(b_t - S_t{Z_t} \right).
\end{align*}

\begin{align*}
R_{X_t|Z_1^T} &= \mathbb{E}_{X_t|Z_1^T} \left[ \left(X_t - \mathbb{E}(X_t|Z_1^T) \right)  \left(X_t - \mathbb{E}(X_t|Z_1^T) \right)^\prime \right]\\
&= p_Z\left(4(1-p_Z)^2 \left( S_t^\prime(b_t - S_tZ_t)(b_t - S_tZ_t)^\prime S_t \right) \right) \\
&+ (1-p_Z)\left(4 {p_Z}^2 \left( S_t^\prime(b_t - S_tZ_t)(b_t - S_tZ_t)^\prime S_t \right) \right) \\
&= \underbrace{4p_Z(1-p_Z)}_{\eta(Z)} S_t^\prime(b_t - S_tZ_t)(b_t - S_tZ_t)^\prime S_t.
\end{align*}

\begin{align*}
D_E &= \mathbb{E}_Z \frac{1}{T} \sum\limits_{t=1}^T \text{tr}\left(R_{X_t|Z_1^T} \right) \\
&= \mathbb{E}_Z \frac{1}{T} \sum\limits_{t=1}^T \text{tr}\left(\eta(Z) S_t^\prime(b_t - S_tZ_t)(b_t - S_tZ_t)^\prime S_t\right) 
\end{align*}

\begin{align*}
&= \mathbb{E}_Z \frac{1}{T} \sum\limits_{t=1}^T \eta(Z) \text{tr}\left( S_t^\prime(b_t - S_tZ_t)(b_t - S_tZ_t)^\prime S_t\right) \\
&= \frac{1}{T} \mathbb{E}_Z  \left[\sum\limits_{t=1}^T \eta(Z) \|S_tZ_t - b_t \|^2 \right] \\
&= \frac{1}{T} \mathbb{E}_Z \left[\sum\limits_{t=1}^T 4 p_Z (1-p_Z)  \|S_tZ_t - b_t \|^2 \right] \\
&= \frac{1}{T} \mathbb{E}_Z \left[\sum\limits_{t=1}^T 4 \frac{f_X(Z) f_X(\tilde{Z})}{(f_X(Z) + f_X(\tilde{Z}))^2}  \|S_tZ_t - b_t \|^2 \right].
\end{align*}

Now, $Z_1^T$ is the transmitted symbols if $X_1^T=Z_1^T$ and key was zero or if $\{X_t=\tilde{Z_t}, \ \forall t \in [T] \}$ and key was one. So $f_Z(Z) = \frac{f_X(Z) + f_X(\tilde{Z})}{2}$. Thus $D_E$,
\begin{align*}
&= \frac{1}{T} \mathbb{E}_Z \left[\sum\limits_{t=1}^T 4 \frac{f_X(Z) f_X(\tilde{Z})}{(f_X(Z) + f_X(\tilde{Z}))^2}  \|S_tZ_t - b_t \|^2 \right] \\
&= \frac{1}{T}\int\!f_Z(Z) \left[\sum\limits_{t=1}^T \frac{4 f_X(Z) f_X(\tilde{Z})}{(f_X(Z) + f_X(\tilde{Z}))^2}  \|S_tZ_t - b_t \|^2 \right]\!dZ \\
&= \frac{1}{T}\int \left[\sum\limits_{t=1}^T \frac{2 f_X(Z) f_X(\tilde{Z})}{f_X(Z) + f_X(\tilde{Z})} \|S_tZ_t-b_t\|^2  \right] \ dZ \\
&= \frac{1}{T} \mathbb{E}_{X} \left[ \sum\limits_{t=1}^T \frac{2f_X(\tilde{X})}{f_X(X) + f_X(\tilde{X})} \|S_tX_t-b_t\|^2\right],
\end{align*}
which proves \eqref{eq::DEE_rough_causal}. 
Again, if we can choose $S_t$'s and $b_t$'s such that,
\begin{align*}
f_X(X) = f_X(\tilde{X}), \ \forall X \in \mathbb{R}^{n T},
\end{align*} 
the distortion $D_E$ becomes,
\begin{align*}
&\frac{1}{T} \mathbb{E}_X \sum\limits_{t=1}^T \|S_tX_t - b_t \|^2  =  \frac{1}{T} \sum\limits_{t=1}^T \mathbb{E}_{X_t} \|S_tX_t - b_t \|^2 \\
= &  \frac{1}{T} \sum\limits_{t=1}^T \text{tr} \left( S_t R_{X_t} S_t^\prime + (b_t - S_t \mu_{X_t}) (b_t - S_t \mu_{X_t})^\prime\right),
\end{align*}
which proves \eqref{eq::DEE_simple_causal}. 
 
\subsection{Proof for Theorem~\ref{thm:gaussian_vector}}
\label{app:gaussian_vector}
	Let the shared key $K$ is $(K_1, K_2, \ldots, K_n)$ where all $K_i$'s are i.i.d. and uniformly distributed in $\{0,1\}^3$.
	Let us also assume that $X = (X^{(1)}, X^{(2)}, \ldots, X^{(n)})$, where each $X^{(i)} \in \mathbb{R}$. Similar to the scheme for scalar case, we create a random vector $V = (V^{(1)}, \ldots, V^{(n)})$ where,
	\begin{align*}
	V^{(i)} & = (X^{(i)} - \mu^{(i)})/\sqrt{\Sigma_{ii}},
	\end{align*}
	and encode $V^{(i)}$ using key $K_i$ as in the case of a scalar for all $i \in [n]$. Thus, the distortion ${D_W} $ will be,
	\begin{align*}
	& =\min\limits_{Z} \text{tr}(R_{X | Z}) = \min\limits_{Z} \sum\limits_{i=1}^{n} \text{Var}(X^{(i)} | Z) \\
	& = \min\limits_{Z} \sum\limits_{i=1}^{n} (\Sigma_{ii}) \text{Var}(V^{(i)} | Z) = \sum\limits_{i=1}^{n} (\Sigma_{ii}) \min\limits_{Z} \text{Var}(V^{(i)} | Z) \\
	& = \sum\limits_{i=1}^{n} (\Sigma_{ii}) \min\limits_{Z^{(i)}} \text{Var}(V^{(i)} | Z^{(i)}) = c \sum\limits_{i=1}^{n} (\Sigma_{ii}) = c~\text{tr} (\Sigma),
	\end{align*}
	where $c = 0.9998$. And since $\text{tr} (\Sigma)$ is the expected distortion even when the adversary has no observations, and as we can not beat this by~\eqref{eq::UpperBoundDW}, this is optimal.

\subsection{Proof for Theorem~\ref{thm:traj}}
\label{app:traj}
\noindent\textbf{Distortion at the adversary's end.} Based on the coding scheme we can see that the adversary get $BU_t + w_t$ by just subtracting $AZ_t$ from $Z_{t+1}$ for $t \in [1:T-1]$. So the adversary's information is given by following set:
\begin{align*}
E_{\text{info}} & = \left\{Z_1, B U_t + w_t, \ t \in [1:T-1] \right\} \\
&= \left\{f(X_1, K), B U_t + w_t, \ t \in [1:T-1] \right\}.
\end{align*} 
Thus, $D(t,Z_1^T) = D(t, E_{\text{info}})  = \text{tr}(R_{{X_t}| E_{\text{info}}}).$

Let's first compute $D(t=1,Z_1^T)$,
\begin{align*}
D(t=1,Z_1^T) & = \text{tr}(R_{{X_1}| E_{\text{info}}}) \stackrel{(a)}{=} \text{tr}(R_{{X_1}| f(X_1, K)})  \stackrel{(b)}{=} c~\text{tr} (\Sigma),
\end{align*}
where (a) is because $U_t$ and $w_t$ are independent on $X_t$ and (b) is due to the encoding used in Theorem~\ref{thm:gaussian_vector} with $c = 0.9998$. 

Now, for other time instances we can use induction to prove that we will have worst case distortion at least $\text{tr}(\Sigma)$.
\begin{align*}
&D(t+1,Z_1^T) = \text{tr}(R_{{X_{t+1}}| E_{\text{info}}}) = \text{tr}(R_{{(AX_t + B U_t + w_t)}| E_{\text{info}}}) \\
& = \text{tr}(R_{{(AX_t)}| E_{\text{info}}})  = \text{tr}(A R_{{X_t| E_{\text{info}}}} A^\prime) = \text{tr}(A^\prime A R_{{X_t} | E_{\text{info}}}) \\
& \stackrel{(a)}{=} \text{tr}(V \Lambda V^\prime R_{{X_t} | E_{\text{info}}}) = \text{tr}(\Lambda V^\prime R_{{X_t}| E_{\text{info}}} V) \\
& \stackrel{(b)}{=} \sum\limits_{i=1}^n \lambda_i d_i(V^\prime R_{X_t| E_{\text{info}}} V ) \stackrel{(c)}{\geq} \sum\limits_{i\in [n]} d_i(V^\prime R_{X_t| E_{\text{info}}} V ) \\
& \stackrel{(d)}{=} \sum\limits_{i\in [n]} \nu_i(V^\prime R_{X_t | E_{\text{info}}} V)  \stackrel{(e)}{=}  \sum\limits_{i\in [n]} \nu_i(R_{X_t | E_{\text{info}}}) \\
& = \text{tr}(R_{X_t | E_{\text{info}}}) \stackrel{(f)}{\geq} c~\text{tr}(\Sigma),
\end{align*} 
where in (a), we do eigenvalue decomposition of $A^\prime A$ which is a positive definite matrix and thus will have non negative eigenvalues; in (b) $d_i(V^\prime R_{X_t| E_{\text{info}}} V )$ is the $i$-th diagonal entry of $V^\prime R_{X_t|E_{\text{info}}} V$; (c) is true because $V^\prime R_{X_t|E_{\text{info}}} V$ is a positive definite matrix and all the diagonal entries of a positive semi definite matrix are non-negative and because of our assumption that singular values of $A$, i.e. the square root of eigenvalues of $A^\prime A$ are all more than one; (d) is because summation of eigenvalues is equal to the sum of all the diagonal entries for any square matrix, where $\nu_i(V^\prime R_{X_t | E_{\text{info}}} V)$ is the $i$-th eigenvalue of $V^\prime R_{X_t | E_{\text{info}}} V$; (e) is because a unitary transform preserve the eigenvalues. Finally (f) is because of the induction step. 